\documentclass{article}
\usepackage{spconf,amsmath,graphicx,mathtools,flexisym,array,float,multirow}
\newcolumntype{P}[1]{>{\centering\arraybackslash}p{#1}}
\usepackage[utf8]{inputenc}
\usepackage[english]{babel}
\usepackage{color}
\usepackage{framed}
\usepackage{cite}
\usepackage{amsmath,amssymb,amsfonts}
\usepackage{graphicx}
\usepackage{upgreek}
\usepackage{textcomp}
\usepackage{cite}
\usepackage{bbold}
\usepackage{ifpdf}
\usepackage{times}
\usepackage{layout}
\usepackage{float}
\usepackage{afterpage}
\usepackage{amsmath}
\usepackage{amstext}
\usepackage{amssymb,bm}
\usepackage{latexsym}
\usepackage{color}
\usepackage{kantlipsum}    

\usepackage{graphicx}
\usepackage{amsmath}
\usepackage{amsthm}
\usepackage{graphicx}

\usepackage{url}
\usepackage[font=scriptsize,caption=false,labelsep=space]{subfig}

\usepackage{booktabs}
\usepackage{multicol}
\usepackage{lipsum}

\usepackage{enumitem}
\usepackage[switch]{lineno}
\usepackage[named]{algo}
\usepackage[noend]{algpseudocode}
\usepackage{algorithm}

\usepackage{seqsplit}

\usepackage{amsmath}
\usepackage{amsthm}
\usepackage{dsfont}
\usepackage{thmtools}
\usepackage{amssymb}
\usepackage[dvipsnames]{xcolor}

\def\bs{\boldsymbol{s}}

\def\bx{\boldsymbol{x}}

\def\bG{\boldsymbol{G}}

\def\balpha{\boldsymbol{\alpha}}

\def\bzeta{\boldsymbol{\zeta}}

\def\btheta{\boldsymbol{\theta}}

\def\bmu{\boldsymbol{\mu}}

\def\bPhi{\boldsymbol{\Phi}}

\def\mba{\mathbf{a}}
\def\mbb{\mathbf{b}}
\def\mbc{\mathbf{c}}
\def\mbd{\mathbf{d}}

\def\mbs{\mathbf{s}}

\def\mbv{\mathbf{v}}
\def\mbw{\mathbf{w}}
\def\mbx{\mathbf{x}}
\def\mby{\mathbf{y}}
\def\mbz{\mathbf{z}}

\def\mbB{\mathbf{B}}
\def\mbC{\mathbf{C}}
\def\mbD{\mathbf{D}}
\def\mbE{\mathbf{E}}

\def\mbG{\mathbf{G}}
\def\mbH{\mathbf{H}}
\def\mbI{\mathbf{I}}

\def\mbP{\mathbf{P}}
\def\mbQ{\mathbf{Q}}
\def\mbR{\mathbf{R}}

\def\mbT{\mathbf{T}}

\def\mbV{\mathbf{V}}
\def\mbW{\mathbf{W}}
\def\mbX{\mathbf{X}}
\def\mbY{\mathbf{Y}}
\def\mbZ{\mathbf{Z}}

\def\bzero{\boldsymbol{0}}

\def\Tr#1{\mathrm{Tr}\left(#1\right)}
\def\vec#1{\mathrm{vec}\left(#1\right)}
\def\diag#1{\mathrm{diag}\left(#1\right)}
\def\Diag#1{\mathrm{Diag}\left(#1\right)}

\def\Re#1{\operatorname{Re}\left(#1\right)}
\def\Im#1{\operatorname{Im}\left(#1\right)}
\def\arg#1{\operatorname{arg}\left(#1\right)}

\newtheorem{theorem}{Theorem}
\newtheorem{proposition}{Proposition}

\newtheorem{remark}{Remark}

\def\T{\top}
\def\H{\mathrm{H}}
\def\j{\mathrm{j}}
\def\vecinv#1{\mathrm{vec}^{-1}\left(#1\right)}

\newcommand*{\rom}[1]{\expandafter\@slowromancap\romannumeral #1@}

\title{Joint Waveform and Passive Beamformer Design in Multi-IRS-Aided Radar}
\name{Zahra Esmaeilbeig$^1$$^\star$,  Arian Eamaz$^1$$^\star$, Kumar Vijay Mishra$^{\dagger}$, and Mojtaba Soltanalian$^\star$ \thanks{$ ~ ^1$Equal contribution. This work was sponsored in part by the National Science Foundation Grant  ECCS-1809225, and in part by the Army Research Office, accomplished under Grant Number W911NF-22-1-0263. The views and conclusions contained in this document are those of the authors and should not be interpreted as representing the official policies, either expressed or implied, of the Army Research Office or the U.S. Government. The U.S. Government is authorized to reproduce and distribute reprints for
Government purposes notwithstanding any copyright notation herein. }
}
\address{$^{\star}$ECE Department, University of Illinois Chicago, Chicago, IL 60607, USA\\
 $^{\dagger}$United States DEVCOM Army Research Laboratory, Adelphi, MD 20783, USA\vspace{-10pt}
	}
\ninept
\begin{document}
\maketitle
\begin{abstract}
Intelligent reflecting surface (IRS) technology has recently attracted a significant interest in non-light-of-sight radar remote sensing. 
Prior works have largely focused on designing single IRS beamformers for this problem. For the first time in the literature, this paper considers multi-IRS-aided multiple-input multiple-output (MIMO) radar and jointly designs the transmit unimodular waveforms and optimal IRS beamformers. To this end, we derive the Cram\'er-Rao lower bound (CRLB) of target direction-of-arrival (DoA) as a performance metric. Unimodular transmit sequences are the preferred waveforms from a hardware perspective. We show that, through suitable transformations, the joint design problem can be reformulated as two unimodular quadratic programs (UQP). To deal with the NP-hard nature of both UQPs, we propose \textit{u}nimodular waveform and \textit{be}amforming design for multi-IRS \textit{r}adar (UBeR) algorithm that takes advantage of the low-cost power method-like iterations. Numerical experiments illustrate that the MIMO waveforms and phase shifts obtained from our UBeR algorithm are effective in improving the CRLB of DoA estimation. 
\end{abstract}
\vspace{-6pt}
\section{Introduction}
\vspace{-6pt}
An intelligent reflecting surface (IRS) is composed of a large array of scattering meta-material elements, which reflect the incoming signal after introducing a pre-determined phase shift~\cite{wu2019towards,ozdogan2019intelligent}. Recently, the benefits of IRS have been investigated for future wireless communications \cite{gong2020toward,hodge2020intelligent,hodge2022index} applications, 
including multi-beam design~\cite{torkzaban2021shaping}, 
 secure parameter estimation~\cite{ahmed2022joint} and joint sensing-communications ~\cite{mishra2022optm3sec,wei2022irs,elbir2022rise}. In this paper, we focus on the IRS-aided radar, where combined processing of line-of-sight (LoS) and non-LoS (NLoS) paths has shown improvement in target estimation and detection \cite{buzzi2022foundations,esmaeilbeig2022irs,song2022intelligent,dehkordi2021reconfigurable} 
through an optimal design of IRS phase shifts.  

Target detection via multiple-input multiple-output (MIMO) IRS-aided  radar was studied  extensively in~\cite{buzzi2022foundations}. In our earlier works on target estimation ~\cite{esmaeilbeig2022irs,esmaeilbeig2022cramer}, we derived the optimal IRS phase shifts based on the mean-squared-error of  the best linear unbiased estimator (BLUE) for complex target  reflection factor~\cite{esmaeilbeig2022irs} and the Cram\'er-Rao lower bound (CRLB) of Doppler  estimation for moving targets~\cite{esmaeilbeig2022cramer}. Recent studies~\cite{song2022intelligent,wang2022stars} focused on  optimization of IRS beamforming based on CRLB of direction-of-arrival (DoA) estimation for a  single  IRS-aided radar. 
More recent works \cite{esmaeilbeig2022cramer,wei2022multi} demonstrate the benefits of deploying multiple IRS platforms instead of a single IRS.

Similar to a conventional radar \cite{radarsignaldesign2022}, a judicious design of transmit waveforms  improves the performance of IRS-aided radar. Whereas designing radar probing signals is a well-studied problem~\cite{radarsignaldesign2022,li2017fast,BOSE2021,hu2016locating,soltanalian2013joint}, it is relatively unexamined for IRS-aided radar. In this context, transmit sequences that  mitigate the non-linearities of amplifiers and yield a uniform power transmission over time are of particular interest. Unimodular sequences with the minimum peak-to-average power ratio exhibit these properties and have been studied in previous non-IRS works for radar applications~\cite{hu2016locating}. In this paper, we jointly design unimodular sequences and IRS beamformers.  

Multipath propagation through multiple IRS platforms increases the spatial diversity of the radar system \cite{xu2021mimo}. 
To this end, we investigate the benefits of multipath processing for multi-IRS-aided target estimation. We first derive the CRLB of DoA estimation for a multi-IRS-aided radar. Then, we formulate the unimodular waveform design problem based on the CRLB minimization for IRS-aided radar as a unimodular quadratic program (UQP). The unimodularity constraint makes the UQP an NP-hard problem. In general, UQP may be relaxed via a semi-definite program (SDP) formulation but the latter has a high computational complexity as well\cite{naghsh2013doppler,9896984}. Inspired by the power method that has the advantage of simple matrix-vector multiplications, \cite{soltanalian2014designing,soltanalian2013joint} proposed \emph{power method like iterations} (PMLI) algorithm to approximate UQP solutions leading to a low-cost algorithm. We formulate the IRS beamforming design as a \emph{unimodular quartic programming} (UQ$^2$P). Prior works \cite{song2015sequence,li2017fast} on unimodular waveform design with good correlation properties 
also lead to UQ$^{2}$Ps, for which they employ a more costly  majorization-minimization technique. On the contrary, we use a quartic to bi-quadratic transformation to solve UQ$^{2}$P by splitting it into two quadratic subproblems. 
Our \emph{u}nimodular waveform and \emph{be}amforming design for multi-IRS \emph{r}adar (UBeR) algorithm is based on the cyclic application of PMLI and provides the optimized CRLB. 
In summary,  the contributions of our work are  introducing the signal model  for a multi-IRS-aided radar system, derivation of the Fisher information for the DoA estimation and developing our  algorithm called  UBeR for  joint Unimodular waveform and beamforming design in multi-IRS-aided radar.

Throughout this paper, we use bold lowercase and bold uppercase letters for vectors and matrices, respectively.  We represent a vector $\mathbf{x}\in\mathbb{C}^{N}$ in terms of its elements $\{x_{i}\}$ as $\mathbf{x}=[x_{i}]^{N}_{i=1}$. The $mn$-th element of the matrix $\mbB$ is $\left[\mbB\right]_{mn}$. The sets of complex and real numbers are $\mathbb{C}$ and $\mathbb{R}$, respectively;  $(\cdot)^{\top}$, $(\cdot)^{\ast}$and $(\cdot)^{\mathrm{H}}$ are the vector/matrix transpose, conjugate and the Hermitian transpose, respectively; trace of a matrix is  $\operatorname{Tr}(.)$; the function $\textrm{diag}(.)$ returns the diagonal elements of the input matrix; and $\textrm{Diag}(.)$  
produces a diagonal/block-diagonal matrix with the same diagonal entries/blocks as its vector/matrices argument. 
 The Hadamard (element-wise) and Kronecker products are $\odot$ and $\otimes$, respectively. The vectorized form of a matrix $\mbB$ is written as $\vec{\mbB}$. The $s$-dimensional all-ones vector, all-zeros vector, and the identity matrix of  size $s\times s$ are $\mathbf{1}_{s}$, $\mathbf{0}_{N}$, and $\mbI_s$, respectively. The  minimum eigenvalue of $\mbB$ is  denoted by $\lambda_{min}(\mbB)$. The real, imaginary, and angle/phase components of a complex number are $\Re{\cdot}$, $\Im{\cdot}$, and $\arg{\cdot}$, respectively.  $\mathrm{vec}_{_{K,L}}^{-1}\left(\mbc\right)$ reshapes the input vector $\mbc\in\mathbb{C}^{KL}$ into a  matrix $\mbC\in\mathbb{C}^{K \times L}$ such that $\vec{\mbC}=c$. 

\vspace{-6pt}
\section{Multi-IRS-Aided Radar System  Model}\label{sec_2}
Consider a 
colocated MIMO radar with $N_t$ transmit and $N_r$ receive antennas, each arranged as uniform  arrays (ULA) with inter-element spacing $d$. The $M$ IRS platforms indexed as IRS$_{1}$, IRS$_{2}$,...,IRS$_{M}$, are implemented at stationary and known locations
, each  equipped with $N_m$ reflecting elements arranged as ULA, with element spacing of $d_m$ between  the antennas/reflecting elements of  IRS$_m$.   
The continuous-time signal transmitted from the $n$-th antenna at time instant $t$ is $x_n(t)$. Denote the $N_t \times 1$ vector of all  transmit signals as $\bx(t)=[x_{_{i}}(t)]_{i=1}^{N_t} \in \Omega^{N_t}$, 
where the set of unimodular sequences is  $\Omega^{n} = \left\{\mbs \in \mathbb{C}^{n}| \mbs= [e^{\textrm{j}\omega_{i}}]^{n}_{i=1}, \omega_i \in [0,2\pi]\right\}$.
The steering vectors of radar transmitter, receiver and the $m$-th IRS are, respectively, 
$\mba_t(\theta)=[1,e^{\textrm{j}\frac{2\pi}{\lambda}d sin \theta},\ldots,e^{\textrm{j}\frac{2\pi}{\lambda}d(N_t-1) sin \theta}]^{\top}$, 
$\mba_r(\btheta)=[1,e^{\textrm{j}\frac{2\pi}{\lambda}d sin \theta},\ldots,e^{\textrm{j}\frac{2\pi}{\lambda}d(N_r-1) sin \theta}]^{\top}$, and 
$\mbb_m(\btheta)= [1,e^{\textrm{j}\frac{2\pi}{\lambda}d_m sin \theta},\ldots,e^{\textrm{j}\frac{2\pi}{\lambda}d_m(N_m-1) sin \theta}]^{\top}$,
where $\lambda$, is the carrier wavelength and   $d$ and $d_m$ are usually assumed to be half the  carrier wavelength. Each reflecting element of IRS$_m$ reflects the incident signal with a phase shift and amplitude change that is configured via a smart controller~\cite{bjornson2022reconfigurable}.   We denote the phase shift vector  of IRS$_m$ by $\mbv_m=[e^{\textrm{j}\phi_{_{m,1}}},\ldots,e^{\textrm{j}\phi_{_{m,N_m}}}]^\T \in \mathbb{C}^{N_m}$, 
where $\phi_{_{m,k}}\in[0,2\pi]$ is the phase shift  associated with the $k$-th passive element of  IRS$_m$. 

Denote the angle between the radar-target, radar--IRS$_m$, and  target-IRS$_m$ by $\theta_{tr}$, $\theta_{ri,m}$, and $\theta_{ti,m}$, respectively. 
Denote target-radar channel by
$\mbH_{tr}=\mba_r(\theta_{tr}) \in \mathbb{C}^{N_r \times 1}$; and radar-target by 
$\mbH_{rt}=\mba_t(\theta_{tr})^{\top}  \in \mathbb{C}^{1\times N_t}$. The LoS or radar-target-radar channel matrix is $\mbH_{rtr}=\mba_r{(\theta_{tr})}\mba_t{(\theta_{tr})}^{\top} \in \mathbb{C}^{N_r \times N_T}$. Analogously,  for the multi-IRS aided radar the NLoS channel matrices associated with IRS$_m$ are defined as $\mbH_{ri,m}=\mbb_m(\theta_{ri,m}) \mba_{t}^{\top}(\theta_{ri,m})\in \mathbb{C}^{N_m \times N_T}$ for radar-IRS$_{m}$; 
$\mbH_{it,m}=\mbb_{m}^{\top}(\theta_{ti,m}) \in \mathbb{C}^{1 \times N_m}$ for IRS$_{m}$-target; 
$\mbH_{ti,m}=\mbb_m(\theta_{ti,m})\in \mathbb{C}^{N_m\times 1}$ for target-IRS$_{m}$; and
$\mbH_{ir,m}=\mba_r(\theta_{ri,m}) \mbb^{\top}_m(\theta_{ri,m}) \in \mathbb{C}^{N_r\times N_m}$ for IRS$_{m}$-radar paths. 

The received signal 
back-scattered from a single target is the superimposition of echoes from both LoS and NLoS paths as \par \noindent\small
 \begin{align}
	    \label{eq::MIMO}
			\mby(t)&=\alpha_{_{rtr}}\mbH_{rtr}\mbx(t-\tau_{rtr})\nonumber\\
			&+\sum_{m=1}^{M}\alpha_{_{ritr,m}}\mbH_{tr}\mbH_{it,m} \bPhi_m \mbH_{ri,m} \mbx(t-\tau_{ritr,m}) \nonumber\\
			&+\sum_{m=1}^{M}\alpha_{_{rtir,m}}  \mbH_{ir,m}\bPhi_m \mbH_{ti,m}\mbH_{rt}\mbx(t-\tau_{rtir,m}) \nonumber\\
			&+ \sum_{m=1}^{M}\alpha_{_{ritir,m}}\mbH_{ir,m}\bPhi_{m}\mbH_{ti,m} \mbH_{it,m}\bPhi_m\mbH_{ri,m}\nonumber\\
			&\mbx(t-\tau_{ritir,m})+\mbw(t),\;\in \mathbb{C}^{N_r},
\end{align}\normalsize
where $\bPhi_{m}=\Diag{\mbv_m}$, $\alpha_{_{{(\cdot),m}}}$ is the complex reflectivity which depends on the  target  back-scattering coefficient and  the atmospheric attenuation, and $\mbw(t)\sim \mathcal{CN}(\bzero,\sigma^2\mbI_{N_t})$ denotes a stationary (homoscedastic) additive  white Gaussian noise (AWGN). 
In general, the received signal may also have an additional inter-IRS interference that should be included while accounting for the SNR. When there is some blockage or obstruction between the radar and target, we have $\alpha_{rtr}\simeq 0$, 
$\alpha_{ritr,m}\simeq 0$ and  $\alpha_{rtir,m}\simeq 0$. We  replace  $\alpha_{_{ritir,m}}$ by $\alpha_{_{m}}$ for  notation brevity. 
The received signal 
becomes\par \noindent 
\begin{align}
 \mby(t)&= \sum_{m=1}^{M}\alpha_{_{m}}\mbH_{ir,m}\bPhi_m\mbH_{ti,m} \mbH_{it,m}\bPhi_m\mbH_{ri,m}\nonumber \\
& \hspace{2cm} \bx(t-\tau_{ritir,m})+\mbw(t).
\end{align}\normalsize
Our goal is to design a radar system for inspecting a  range cell  located at distance $d_{tr}$ with respect to (w.r.t.) the radar transmitter/receiver for a potential target. Assume that the relative time gaps  between any two multipath signals are very small in comparison to the actual roundtrip delays, i.e., $\tau_{ritir,m} \approx \tau_0 =\frac{2d_{tr}}{c}$ for $m \in \{1,\ldots,M\}$, where $c$ is the speed of light. We collect  $N$  slow-time samples at the rate $1/T_s$  from the signal, at $t=nT_s$, $n=0,\ldots,N-1$. Hence, corresponding to  the  range-cell of interest, the received signal vector is \par \noindent\small 
\begin{equation}
\label{eq_1}
\mby[n]= \sum_{m=1}^{M} \alpha_{_{m}} \mbH_{m} \mbx[n]+\mbw[n],\quad \mby[n] \in \mathbb{C}^{N_r \times 1} , 
\end{equation}\normalsize
where 
$\mbx[n]=\mbx(\tau_0+nT_s)\in \mathbb{C}^{N_t \times 1}$, $\mby[n]=[y_{_{i}}[n]]_{i=1}^{N_r}$, and we define $\mbH_m=\mbH_{ir,m}\bPhi_{m}\mbH_{ti,m} \mbH_{it,m}\bPhi_m\mbH_{ri,m}\in\mathbb{C}^{N_r\times N_t}$. The delay $\tau_0$ is aligned on-the-grid so that $n_0=\tau_0/T_s$ is an integer~\cite{mishra2017sub}.

Collecting all discrete-time samples for $N_r$  receiver antennas, the received signal is the $N_r \times N$  
matrix
 $\mbY=[\mby[0],\ldots,\mby[N-1]]=\sum_{m=1}^{M} \alpha_{_{m}}  \mbH_m \mbX+\mbW,$
where \small$\mbX=[\mbx[0],\ldots,\mbx[N-1]]\in \mathbb{C}^{N_t \times N}$\normalsize 
and \small$\mbW$ $=[\mbw[0],\ldots,$ $\mbw[N-1]]\in \mathbb{C}^{N_r \times N}$\normalsize. 
Vectorizing 
as $\mby=\vec{\mbY}$ yields\par\noindent
\begin{equation}\label{eq:vec_y}
\mby= \sum_{m=1}^{M} \alpha_{_{m}}  \vec{\mbH_m \mbX} +\vec{\mbW} =\tilde{\mbX}\tilde{\mbH}\balpha +\tilde{\mbw} ,
\end{equation}\normalsize
where $\tilde{\mbX}=\mbX^{\T}\otimes \mbI_{_{N_r}}$, $\tilde{\mbH}=[\tilde{\mbH}_{_1},\ldots,\tilde{\mbH}_{_M}]$, $\tilde{\mbH}_m=\vec{\mbH_m} $ for $m \in \{1,\ldots,M\}$, $\tilde{\mbw}=\vec{\mbW}$ and $\balpha=[\alpha_m]_{m=1}^M$. Given that $\mbw(t)$ is AWGN in~\eqref{eq::MIMO}, it is  easily observed  that $\mby \sim \mathcal{CN}(\bmu,\mbR)$, where $\bmu=\tilde{\mbX}\tilde{\mbH}\balpha$ and $\mbR=\sigma^2\mbI_{_{N_rN}}$. Note that, since   $\mbw(n)$ is a stationary process and i.i.d. with $\sigma^{2}$ variance, through vectorization and stacking all ensembles as one vector, the resulting process is still stationary and i.i.d with the same variance.


Our goal is to show the effectiveness of placing $M$ IRS platforms in  estimating the DoA of the target in the LoS path,  i.e.  $\theta_{tr}$. For simplicity, we  consider a two-dimensional (2-D) scenario, where the radar, IRS platforms and the target are in the same plane. Our analysis can be easily extended to 3-D scenarios. The following remark states that  the estimation of  DoAs in the NLoS paths, $\theta_{ti,m}$, for  $m \in \{1,\ldots,M\}$ is  equivalent to an estimation of $\theta_{tr}$. 
\begin{remark}\label{rmk:one}
Estimation of the  vector of target DoAs, 
$\bzeta=[\theta_{ti,1},$ $\ldots,\theta_{ti,M}]^\T$ 
is equivalent to estimating scalar DoA  parameter, $\theta_{tr}$. 
This follows because, given the  locations of radar, IRS platforms and  potential target range in the 2-D plane, we have 
$\bzeta=[\theta_{tr}+\theta_{1},\ldots,\theta_{tr}+\theta_{_{M}}]^\T$,
 where  $\theta_{m}$ for $m \in \{1,\ldots,M\}$ are known. 
 \end{remark}
\section{UQP-Based CRLB Optimization}\label{sec_3}
For an unbiased estimator of a parameter $\theta_{tr}$ ($\theta$, hereafter), the variance  of  $\hat{\theta}$ is lower bounded as $\mbE \{{(\hat{\theta}-\theta)(\hat{\theta}-\theta)^H}\}\geq \textrm{CRLB}(\theta)$~\cite{kay1993fundamentals}. Theorem~\ref{theorem1} below unveils the Fisher information $F_{\theta}=\left(\textrm{CRLB}(\theta)\right)^{-1}$. 
\begin{theorem} 
\label{theorem1}
Consider for the multi-IRS-aided radar, the receive signal model in~\eqref{eq:vec_y}. The Fisher information of LoS DoA 
$\theta$ is  
\par\noindent\small
\begin{equation}\label{eq:fisher}
F_{\theta}=\frac{2}{\sigma^2}\Re{\balpha^{\H}\dot{\tilde{\mbH}}^{\H}\tilde{\mbX}^{\H}\tilde{\mbX}\dot{\tilde{\mbH}}\balpha},
\end{equation}\normalsize
where $\dot{\tilde{\mbH}}=\left[\dot{\tilde{\mbH}}_{_1},\ldots,\dot{\tilde{\mbH}}_{_M}\right]$, $\dot{\tilde{\mbH}}_{m}=\vec{\dot{\mbH}_{_m}}$ and 
$\dot{\mbH}_m=\frac{\partial\mbH_{_m}}{\partial\theta}=b_m\mbH_{ir,m}\bPhi_{m}(\mbb_m(\theta_{ti,m})(\mbd \odot\mbb_m(\theta_{ti,m}))^{\T}$ 
$+(\mbd \odot \mbb_m(\theta_{ti,m}))\mbb_m(\theta_{ti,m})^{\T})\bPhi_m\mbH_{ri,m}$, 
with NLoS DoAs 
being 
$\theta_{ti,m}=\theta+\theta_{m}$, $b_m=\j \frac{2\pi d_m}{\lambda}\cos(\theta_m)$ and $\mbd=[0,\ldots,N_m-1]^{\T}$.\vspace{-8pt}
\end{theorem}
\begin{proof}
 Given the observations $\mby\sim \mathcal{CN}(\bmu(\theta),\mbR)$, using Slepian-Bangs formula \cite[Chapter 3C]{kay1993fundamentals}, the Fisher information is\par\noindent\small
\begin{align}\label{eq:Slepian}
F_{\theta}=\Tr {\mbR^{-1}\frac{\partial \mbR }{\partial\theta} \mbR^{-1}\frac{\partial \mbR }{\partial\theta}}+
2\Re{\frac{\partial\bmu(\theta)}{\partial\theta}^H\mbR^{-1}\frac{\partial\bmu(\theta)}{\partial\theta}}.
\end{align}\normalsize
From \eqref{eq:vec_y}, $\bmu(\theta) = \tilde{\mbX}\tilde{\mbH}\balpha$. Also, given the  above mentioned definitions, we have $\frac{\partial \mu(\theta)}{\partial \theta}=\tilde{\mbX}\dot{\tilde{\mbH}}\balpha$. Substituting this in~\eqref{eq:Slepian} and using $\mbR=\sigma^2\mbI$, one arrives at \eqref{eq:fisher}.
\end{proof}
\begin{remark}
In the  absence of the IRS, 
\textit{ceteris paribus}, the LoS Fisher information 
is 
$F_{\theta}=\frac{2|\alpha_{rtr}|^2 }{\sigma^2}\|\tilde{\mbX}\dot{\tilde{\mbH}}_{rtr}\|_2^2$, 
where 
$\dot{\tilde{\mbH}}_{rtr}=\vec{\dot{\mbH}_{_{rtr}}}$, 
$\dot{\mbH}_{rtr}=\frac{\partial\mbH_{_{rtr}}}{\partial\theta}=\j \frac{2\pi d}{\lambda}\cos(\theta)\left( \left(\mbd' \odot\mba_r{(\theta)}\right)\mba_t{(\theta)}^{\top} \right.$ 
$\left.+\mba_r{(\theta)}(\mbd' \odot \mba_t{(\theta)}^{\top} )\right)$,
$N_r=N_t$, and $\mbd'=[0,\ldots,N_r-1]^{\T}$. \vspace{-8pt}
\end{remark}
To design the unimodular waveform $\mbX$, the following proposition casts 
$F_{\theta}$ in standard quadratic form.
\begin{proposition}\label{Neg_prop1} 
 The Fisher information $F_{\theta}$ of LoS DoA is \par\noindent\small
\begin{align}
\label{eq:neg_crlb1}
F_{\theta}\left(\mbX\right)=\vec{\mbX}^{\H}(\mbI_N \otimes \mbB)^{\H}(\mbI_N \otimes \mbB)\vec{\mbX},
 \end{align}\normalsize
 where $\mbB=\frac{\sqrt{2}}{\sigma}\mathrm{vec}_{_{_{N_r,N_t}}}^{-1}\left(\dot{\tilde{\mbH}}\alpha\right)\in \mathbb{C}^{N_r\times N_t}$.\vspace{-8pt}
\end{proposition}
\begin{proof}\label{Neg_prof3}  
Given  $\tilde{\mbX}=\mbX^{\T}\otimes \mbI_{_{N_r}}$, rewrite Fisher information in (\ref{eq:fisher}) as\par\noindent\small
\begin{equation}\label{Negr_A1}
 F_{\theta}=\frac{2}{\sigma^{2}}\operatorname{Re}\left\{\left(\left(\mbX^{\top}\otimes \mbI_{N_r}\right)\dot{\tilde{\mbH}}\alpha\right)^{\mathrm{H}}\left(\left(\mbX^{\top}\otimes \mbI_{N_r}\right)\dot{\tilde{\mbH}}\alpha\right)\right\}.
\end{equation}\normalsize
Since the argument of real operator is a real number, we can put it out of the real operator. Using the identity \small$\left(\mbX^{\top}\otimes \mbI_{N_r}\right)\dot{\tilde{\mbH}}\alpha=\left( \mbI_{N}\otimes \operatorname{vec}_{_{N_r,N_t}}^{-1}\left(\dot{\tilde{\mbH}}\alpha\right)\right)\vec{\mbX}$ \normalsize in~\eqref{Negr_A1}, we immediately get \eqref{eq:neg_crlb1}.
\end{proof}
Using the expression in~\eqref{eq:neg_crlb1},  
we recast the unimodular waveform design objective as a unimodular quadratic objective that leads to a UQP. To proceed with IRS beamformer design, define, \small
$\dot{\tilde{\mbH}}_{_m}= \mbD_m \vec{\mbV_m}$,
$\mbD=\Diag{\mbD_1,\ldots,\mbD_m}$, \par\noindent\small
\begin{equation}\label{eq:V_m}
\mbD_m=\left(\mbC_m^{\T}\diag{\mbd}\otimes \mbC_m^{\T}\right)+\left(\mbC_m^{\T}\otimes\mbC_m^{\T}\diag{\mbd}\right)
\end{equation}\normalsize
and
$\mbC_m=\Diag{\mbb_m(\theta_{ti,m})}\mbH_{ri,m}$,
where the unimodular phase shifts for IRS$_m$ are given by $\mbv_m=\diag{\bPhi_m}$ or $\mbV_m=\vec{\mbv_m\mbv_m^{\top}}$. In order to obtain~\eqref{eq:V_m}, we imposed the reciprocity, $\mbH_{ir,m}=\mbH_{ri,m}^{\T}$ for a radar with collocated antennas and $N_r=N_t$. 
For the IRS beamforming, the Fisher information $F_{\theta}$ w.r.t. phase shifts is recast in the following proposition.
\begin{proposition}\label{Neg_prop2}
 The Fisher information 
 is quartic in phase shifts: 
\begin{equation}\label{eq_neg10}
F_{\theta}(\upnu)=\upnu^{\H}\mbQ^{\H}_{1}(\upnu)\mbT \mbQ_{1}(\upnu) \upnu=\upnu^{\H}\mbQ_{2}^{\H}(\upnu)\mbT \mbQ_{2}(\upnu) \upnu,
\end{equation}
where 
 $\upnu = \left[\mbv^{\top}_{1},\mbv^{\top}_{2}, \cdots ,\mbv^{\top}_{M}\right]^{\top}\in \mathbb{C}^{MN_{m}}$, 
 $\mbT =\mbD^{H}\mbP^{\H}\mbZ^{\ast}\mbP\mbD$, 
 $\mbQ_{1}(\upnu) =\Diag{[\mbv_1\otimes \mbI_{N_m},\ldots,\mbv_{_{M}}\otimes \mbI_{N_m}]}$, 
 $\mbQ_{2}(\upnu) = \Diag{[ \right.$
 $\left. \mbI_{N_m}\otimes\mbv_1,\ldots,\mbI_{N_m}\otimes\mbv_{_M}]}$, 
 $\mbZ=(\mbI_{N_rN_t} \otimes  \balpha^{*}\balpha^{\T})^{\T}(\tilde{\mbX}^{\T}\tilde{\mbX}^{*}\otimes  \mbI_{M})$, 
and $\mbP$ is the commutation matrix, i.e., $\vec{\dot{\tilde{\mbH}}^{\T}}=\mbP\vec{\dot{\tilde{\mbH}}}$.\vspace{-8pt}
\end{proposition}
\begin{proof}
 \label{Neg_prof1}
The Fisher information is 
$F_{\theta}= \Tr{\dot{\tilde{\mbH}}\balpha\balpha^{\H}\dot{\tilde{\mbH}}^{\H}\tilde{\mbX}^{\H}\tilde{\mbX}} 
= \vec{\balpha^{*}\balpha^{\T}\dot{\tilde{\mbH}}^{\T}}^{\T}
\vec{\dot{\tilde{\mbH}}^{\H}\tilde{\mbX}^{\H}\tilde{\mbX}} 
=\vec{\dot{\tilde{\mbH}}^{\T}}^{\T}\mbZ \vec{\dot{\tilde{\mbH}}^{\H}}$. 
Since $F_{\theta}$
is real, we have 
$F_{\theta}\left(\mbz\right)=\vec{\dot{\tilde{\mbH}}^{\T}}^{\H}\mbZ^{\ast} \vec{\dot{\tilde{\mbH}}^{\T}} 
=\vec{\dot{\tilde{\mbH}}}^{\H}\mbP^{\H}\mbZ^{\ast}\mbP\vec{\dot{\tilde{\mbH}}} 
=\mbz^{\H}\mbD^{H}\mbP^{\H}\mbZ^{\ast}\mbP\mbD \mbz$, 
where $\mbz=\vec{[\mbV_1,\ldots,\mbV_M]}=[\vec{\mbv_m\mbv_m^{\top}}^{\top},\ldots,\vec{\mbv_m\mbv_m^{\top}}^{\top}]^{\top}$. Applying the identity 
\begin{equation}
\label{Neg_Ar10}
\operatorname{vec}\left(\mbv_{m}\mbv^{\top}_{m}\right)=\left(\mbI_{N_m}\otimes\mbv_{m}\right)\mbv_{m}=\left(\mbv_{m}\otimes\mbI_{N_m}\right)\mbv_{m},  
\end{equation}
yields $\mbz=\mbQ_{1}(\upnu)\upnu=\mbQ_{2}(\upnu)\upnu$. This completes the proof.
\end{proof}
To jointly optimize $\mbv_m=\diag{\bPhi_m}$ 
and 
$\mbX$,  
we solve 
\par\noindent\small
\begin{align}
\label{eq:opt_crlb}
\underset{\mbX\in\Omega^{N_{t}\times N}, \upnu\in \Omega^{M N_{m}}}{\textrm{maximize}} \quad  F_{\theta}(\upnu,\mbX),
\end{align}\normalsize
which leads to the CRLB minimization. 
Note that this problem is UQP w.r.t. $\mbX$ but quartic or UQ$^{2}$P w.r.t. the phase shifts $\upnu$. 


\vspace{-8pt}
\section{UBeR Algorithm}\label{sec_4} 
\vspace{-8pt}
We resort to a task-specific \emph{alternating optimization} (AO) or \emph{cyclic algorithm} \cite{bezdek2003convergence,tang2020polyphase,soltanalian2013joint}, wherein we optimize~\eqref{eq:opt_crlb} for $\mbX$ and $\upnu$ cyclically. To tackle each subproblem, we adopt power method-like iterations (PMLI) \cite{soltanalian2014designing}, which is a computationally efficient procedure to tackle the UQP. The PMLI resembles the well-studied \emph{power method} for computing the dominant eigenvalue/vector pairs of matrices \cite{soltanalian2014designing}. Given a matrix $\mbG$, the following problem is a UQP\cite{soltanalian2014designing}:
\par\noindent\small
\begin{equation}\label{eq:opt_crlb2}
\mathcal{P}_{1}:~\underset{\mbs \in \Omega^{n}}{\textrm{maximize}} \quad\mbs^{\H} \mbG \mbs.
\end{equation}\normalsize
If $\mbG$ is positive semidefinite, the PMLI iterations\par\noindent\small
\begin{equation}\label{eq:UQP_it}
\mbs^{(t+1)}=e^{\textrm{j}\arg{\mbG\mbs^{(t)}}},
\end{equation}\normalsize
lead to a monotonically increasing objective value for the UQP.

\noindent\textit{\textbf{\underline{Unimodular Waveform Design:}}} From Proposition~\ref{Neg_prop1}, the Fisher information $F_{\theta}$ for the unimodular waveform $\mbX$ is the unimodular quadratic objective in 
(\ref{eq:neg_crlb1}). Let $\bs=\vec{\mbX}$, and $\mbG=(\mbI_N \otimes \mbB)^{\H}(\mbI_N \otimes \mbB)$.
Therefore, the vectorized $\mbX$ is obtained from $\mathcal{P}_{1}$ via the iterations ($t\geq 0$):\par\noindent\small
\begin{equation}\label{eq:UQP_it2}
\vec{\mbX^{(t+1)}}=e^{\textrm{j}\arg{(\mbI_N \otimes \mbB)^{\H}(\mbI_N \otimes \mbB)\vec{\mbX^{(t)}}}}.
\end{equation}\normalsize
If $\mbG$ is not positive semidefinte, at each iteration we use the diagonal loading technique, i.e., $\tilde{\bG} \leftarrow \mbG+\lambda_{m}\mbI$, where the loading parameter
$\lambda_{m} \geq -\lambda_{min}(\mbG)$. Note that diagonal loading with $\lambda_{m}\mbI$ has no effect on the solution of (\ref{eq:UQP_it2}) because $\mbs^{\mathrm{H}} \tilde{\mbG}(\bs) \mbs=\lambda_{m}N_{_{t}}N+\mbs^{\mathrm{H}} \mbG \mbs$.
\textit{\textbf{\underline{IRS Beamforming Design:}}} For the phase shifts optimization, we find an alternative bi-quadratic formulation to the quartic $F_{\theta}\left(\upnu\right)$. Define
\vspace{-6pt}\par\noindent\small
\begin{equation}
\label{eq_Neg_5}
g\left(\upnu_{1},\upnu_{2}\right)= \frac{1}{2}\left(\upnu^{\H}_{1}\mbG_{1}\left(\upnu_{2}\right) \upnu_{1}+\upnu^{\H}_{2}\mbG_{1}\left(\upnu_{1}\right) \upnu_{2}\right), 
\end{equation}\normalsize
where $\mbG_{1}(\upnu)=\mbQ^{\H}_{1}(\upnu)\mbT \mbQ_{1}(\upnu)$. The function $g(.,.)$ is  symmetric, i.e., $g\left(\upnu_1, \upnu_2\right)= g\left(\upnu_2, \upnu_1\right)$ and from proposition~\ref{Neg_prop2}, $F_{\theta}(\upnu)=g(\upnu,\upnu)$. According to (\ref{Neg_Ar10}), one can readily verify that $\mbQ_{1}(\upnu_{1})\upnu_{2}=\mbQ_{2}(\upnu_{2})\upnu_{1}$. As a result, $g\left(\upnu_{1},\upnu_{2}\right)$ is rewritten as
\par\noindent\small
\begin{equation}
\label{eq_Neg_6}
g\left(\upnu_{1},\upnu_{2}\right)=\mbv^{\H}_{1}\mbE\left(\upnu_{2}\right) \upnu_{1}, 
\end{equation}\normalsize
where $\mbE\left(\upnu_{2}\right)=\frac{\left(\mbG_{1}\left(\upnu_{2}\right)+\mbG_{2}\left(\upnu_{2}\right)\right)}{2}$, and $\mbG_{2}(\upnu)=\mbQ^{\H}_{2}(\upnu)\mbT \mbQ_{2}(\upnu)$.
Fixing either $\upnu_{1}$ or $\upnu_{2}$ and minimizing $g\left(\upnu_1,\upnu_2\right)$ w.r.t. the other variable requires solving the following UQP:\par\noindent\small
\vspace{-10pt}
\begin{align}\label{neg8}
\underset{\upnu_{j}\in\Omega^{MN_{m}}}
{\textrm{minimize}} \quad \upnu_{j}^{\mathrm{H}} \tilde{\mbE}(\upnu_{i}) \upnu_{j}, \quad i\neq j \in \left\{1,2\right\},
\end{align}\normalsize 
where we used the diagonal loading, $\tilde{\mbE}(\upnu_{i}) \leftarrow \lambda_{M} \mbI - \mbE(\upnu_{i})$, with  
$\lambda_{M}$ being the maximum eigenvalue of $\mbE(\upnu_{i})$. Note that diagonal loading has no effect on the solution  because $\upnu_j^{\mathrm{H}} \tilde{\mbE}(\upnu_i)\upnu_j=\lambda_{M}MN_m-\upnu_j^{\mathrm{H}} \mbE(\upnu_i) \upnu_j$. Moreover,  $\tilde{\mbE}(\upnu_{i})$ is positive semidefinite, to satisfy the requirement of PMLI.

To guarantee that the maximization of $g\left(\upnu_1, \upnu_2\right)$ w.r.t. $\upnu_{1}$ and $\upnu_{2}$ also maximizes $F_{\theta}\left(\upnu\right)$, a regularization would be helpful.
Therefore, we add the norm-$2$ error between $\upnu_{1}$ and $\upnu_{2}$ as a \emph{penalty} function to (\ref{neg8}), we obtain\par\noindent\small
\begin{equation}\label{neg11}
\underset{\upnu_{j}\in\Omega^{MN_{m}}}
{\textrm{minimize}} \quad \upnu_{j}^{\mathrm{H}} \tilde{\mbE}(\upnu_{i}) \upnu_{j} + \eta ||\upnu_{i} - \upnu_{j}||_2^2 , \quad i\neq j \in \left\{1,2\right\},
\end{equation}\normalsize
where $\eta$ is Lagrangian multiplier. Rewrite the objective of (\ref{neg11}) as  
$\bar\upnu_{j}^{\mathrm{H}}\underbrace{\begin{bsmallmatrix} \tilde{\mbE}(\upnu_{i}) &-\eta\upnu_{i} \\ -\eta \upnu_{i}^{\mathrm{H}} & 2\eta M N_{m} \end{bsmallmatrix}}_{\mathcal{E}(\upnu_{i})}\bar\upnu_{j}$, 
where $\bar\upnu_{j}=\left[\upnu^{\top}_{j} \quad 1\right]^{\top}$. Then, UQP for
(\ref{eq:opt_crlb}) w.r.t. $\upnu$ becomes\par\noindent\small
\begin{equation}\label{neg13}
   \mathcal{P}_{2}:~  \underset{\upnu_{j}\in\Omega^{MN_{m}}}
{\textrm{maximize}} \quad
\bar\upnu_{j}^{\mathrm{H}}
 \underbrace{\begin{bmatrix} \hat\lambda_{M}\mbI-\tilde{\mbE}(\upnu_{i}) &\eta\upnu_{i} \\ \eta \upnu_{i}^{\mathrm{H}} & \hat\lambda_{M}-2\eta M N_{m} \end{bmatrix}}_{=\hat{\mbE}(\upnu_{i}) =\hat\lambda_{M}\mbI-\mathcal{E}(\upnu_{i})}
 \bar\upnu_{j},   
\end{equation}\normalsize
where 
$\hat\lambda_{M}$ is the maximum eigenvalue 
of $\mathcal{E}(\upnu_{i})$. 

To tackle the UQ$^{2}$P for maximizing $F_{\theta}$, we solve the \emph{bi-quadratic} program (\ref{neg13}) using PMLI in~\eqref{eq:UQP_it}. 
Algorithm~\ref{algorithm_1} summarizes the proposed steps. The PMLI in UBeR have previously been shown to be convergent in terms of both the optimization objective and variable \cite{soltanalian2013joint,hu2016locating}. 

\begin{algorithm}[H]
\caption{\emph{U}nimodular waveform and \emph{be}amforming design for multi-IRS-aided \emph{r}adar (UBeR). }
    \label{algorithm_1}
    \begin{algorithmic}[1]
    \Statex \textbf{Input:} Initialization values  $\mbX^{(0)}$ and $\upnu_1^{(0)}$ and $\upnu_2^{(0)}$, Lagrangian multiplier $\eta$, total number of iterations $\Gamma_{1}$ and $\Gamma_{2}$ for problems $\mathcal{P}_1$ and $\mathcal{P}_2$, 
    respectively.
    \Statex \textbf{Output:} Optimized phase shifts $\upnu^*$, unimodular waveform $\mbX^*$. 
    \vspace{0.5mm}
     \State Obtain $F_{\theta}\left(\mbX\right)$ and $F_{\theta}\left(\upnu\right)$ from (\ref{eq:neg_crlb1}) and (\ref{eq_neg10}), respectively.
     \State 
     $\mbB\gets\vecinv{\dot{\tilde{\mbH}}\balpha}\in \mathbb{C}^{N_r\times N_t}$,  
     \State $\mbG\gets\left(\mbI_N \otimes \mbB\right)^{\H}\left(\mbI_N \otimes \mbB\right)$.
    
     \vspace{1.2mm}
    \For{$t_{1}=0:\Gamma_{1}-1$} $\triangleright$ Update the unimodular waveform
    \vspace{1.2mm}
    \For{$t_{2}=0:\Gamma_{2}-1$} 
    $\triangleright$ Bi-quadratic programming 
    via PMLI \small
    \vspace{1.2mm}
    \State $\upnu^{(t_{2}+1)}_{1}\gets e^{ \textrm{j} \operatorname{arg}\left(\begin{bmatrix} \mbI_{_{M N_{m}}} \mathbf{0}_{_{MN_{m}}}\end{bmatrix}\hat{\mbE}\left(\upnu^{(t_{2})}_{2},\mbX^{(t_{1})}\right)\bar\upnu^{(t_{2})}_{1}\right)},\label{neg170}$
\vspace{1.2mm}
\State 
$\upnu^{(t_{2}+1)}_{2}\gets e^{\textrm{j} \operatorname{arg}\left(\begin{bmatrix} \mbI_{_{M N_{m}}} \mathbf{0}_{_{M N_{m}}}\end{bmatrix}\hat{\mbE}\left(\upnu^{(t_{2}+1)}_{1},\mbX^{(t_{1})}\right)\bar\upnu^{(t_{2})}_{2}\right)}.\label{neg160}$
\normalsize
    \EndFor
    \vspace{1.2mm}
     \State  
      $\upnu^{(t_{1}+1)}\leftarrow$ $\upnu^{(\Gamma_{2})}_{1}$ or $\upnu^{(\Gamma_{2})}_{2}$.
      \vspace{1.2mm}
      \State 
      $\vec{\mbX^{(t_{1}+1)}}\gets e^{\textrm{j}\arg{\mbG\left(\upnu^{(t_{1}+1)}\right)\vec{\mbX^{(t_{1})}}}}$.
    \EndFor
    \vspace{1.2mm}
    \State \Return  
      $\left\{\upnu^{\star}, \mbX^{\star}\right\}\leftarrow \left\{\upnu^{(\Gamma_{1})}, \mbX^{(\Gamma_{1})}\right\}$.
    \end{algorithmic}
\end{algorithm}
\section{Simulation results}\label{sec_5}
\begin{figure}[t]
\centering
	\includegraphics[width=1.01\columnwidth]{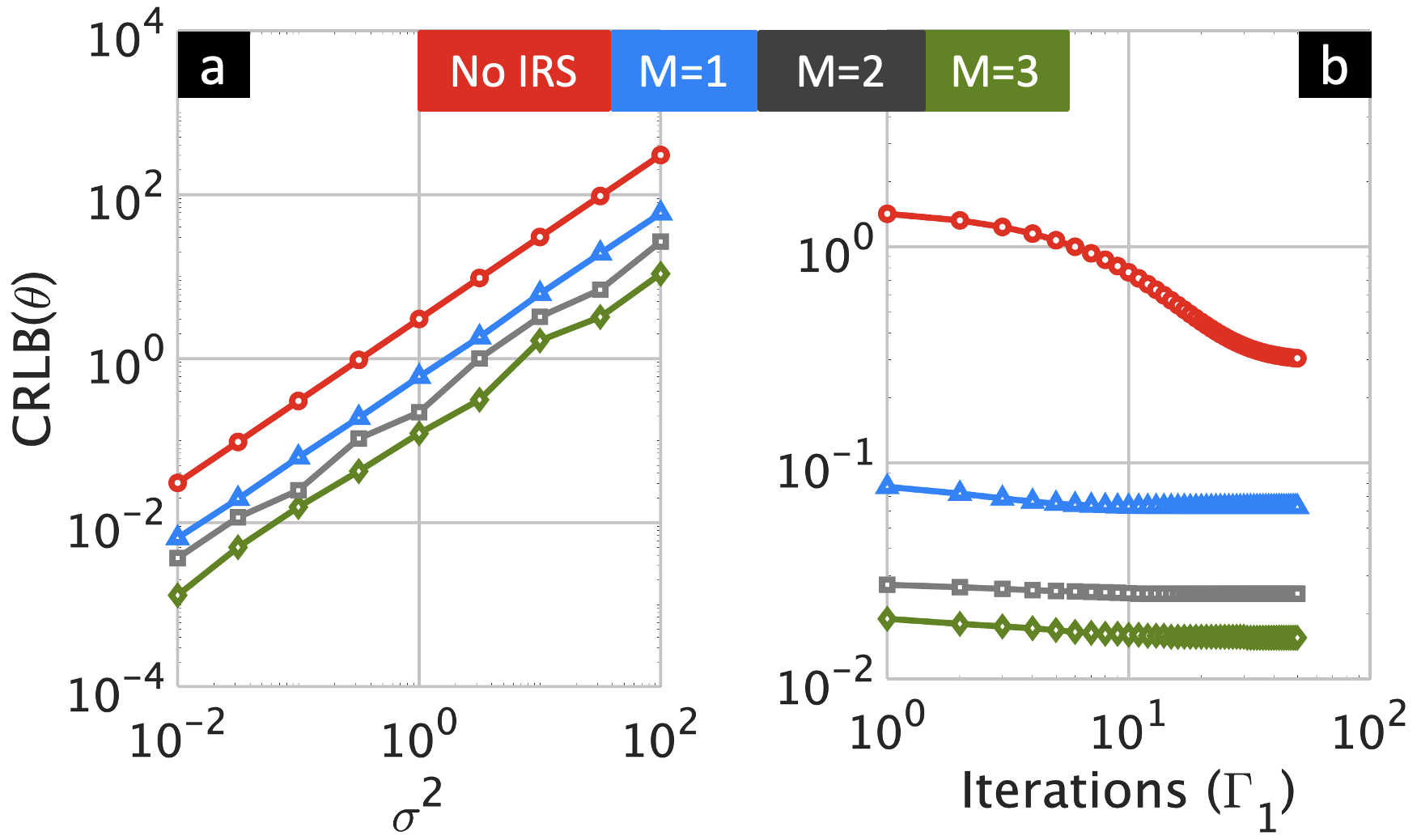}
	\caption{The optimized CRLB of DoA versus (a) $\sigma^2$ for fixed number of iterations $\Gamma_1 =50$, and (b) $\Gamma_1$ for fixed $\sigma^2=0.1$. In both experiments, $\Gamma_2$ was set to $20$. 
	}
 \label{fig_N}
\end{figure}
We consider a radar, equipped with  $N_r=N_t=8$ antennas for transmitter/receiver, positioned in the 2-D Cartesian plane at [$0$ m, $0$ m], sensing a target at [$5000$ m, $5000$ m]. We placed three IRS platforms with $N_m=8$ reflecting elements arranged as ULA  with the first elements located at [$500$ m, $500$ m], [$500$ m, $-800$ m], and [$300$ m, $1300$ m]. The IRS platforms were deployed at far ranges w.r.t. the radar. 
Usually, distant targets tend to have obstructed or very weak LoS signal. In these cases, the received signal from the NLoS paths, i.e., the signal propagated through IRS platforms, is helpful in boosting the reflected LoS echo strength. 

For a point-like target, given the  radar, target and IRS  positions,  the corresponding  radar--IRS$_{m}$ and  target--IRS$_m$ angles $\theta_{ir,m}$ and $\theta_{ti,m}$, for $m \in \{1,\ldots,M\}$, are obtained through geometric computations. The complex reflectivity coefficients $\{\alpha_{m}\}$, which correspond to a Swerling 0 target model are generated from a $\mathcal{CN}(0,1)$. In Algorithm 1, we set $\Gamma_1=50$ and $\Gamma_2=20$ for all iterations. Throughout all our experiments  the Lagrangian  multiplier $\eta$  is  tuned to $0.1$. Initially, all IRS platforms are  set to impose zero phase shift $\upnu_i^{(0)}=\bzero_{_{MN_m}}$ for $i\in \{1,2\}$. The number of slow-time samples  is set to $N=50$ and the samples in  $\mbX^{(0)}$ are generated from a normal distribution. Fig.~\ref{fig_N}a illustrates that the multiple IRS-aided radar outperforms the single-IRS aided radar. 
Further, Fig.~\ref{fig_N}b indicates that iterations of Algorithm 1 result in a monotonically decreasing CRLB.
\vspace{-8pt}
\section{Summary}\label{sec:conclusion}
\vspace{-8pt}
Waveform design for IRS-aided radar is relatively unexplored in prior works. In this context, this paper studies a new set of waveform design problems. 
Numerical experiments demonstrate 
that the deployment of multiple IRS platforms leads to a better achievable estimation performance 
compared to non-IRS and single-IRS systems. Some IRS model enhancements that should be accounted for in the future include the inter-IRS interference and quantization of the IRS phases.
\bibliographystyle{IEEEtran}
\bibliography{refs}
\end{document}